\documentclass{article}
\usepackage[nohyperref, accepted]{icml2012}

\usepackage{color}
\usepackage{url}
\usepackage{verbatim} % allows multiline comments
\usepackage{subfigure} % for the subfloat in the figs
\usepackage{multirow}
\usepackage{algorithm}
\usepackage{algorithmic}
\usepackage{times}
\usepackage{xspace}
\usepackage{graphicx}
\usepackage{epsfig}
\usepackage{amsmath}
\usepackage{amsthm}
\usepackage{amssymb}
\usepackage{natbib}
\usepackage{ctable}

\newcommand{\nedge}{k}
\newcommand{\bt}{\mathbf{t}}

\newcommand{\hide}[1]{}
\newcommand{\xhdr}[1]{\vspace{1.7mm}\noindent{{\bf #1.}}}

\newtheorem{theorem}{Theorem}

\newcommand{\ie}{\emph{i.e.}}
\newcommand{\expo}{{\textsc{Exp}}\xspace}
\newcommand{\pow}{{\textsc{Pow}}\xspace}
\newcommand{\ray}{{\textsc{Ray}}\xspace}

\newcommand{\netinf}{{\textsc{Net\-Inf}}\xspace}
\newcommand{\connie}{{\textsc{ConNIe}}\xspace}
\newcommand{\netrate}{{\textsc{Net\-Rate}}\xspace}
\DeclareMathOperator*{\argmax}{argmax}

\begin{document}

\icmltitlerunning{Submodular Inference of Diffusion Networks from Multiple Trees}

\twocolumn[

\icmltitle{Submodular Inference of Diffusion Networks from Multiple Trees}

\icmlauthor{Manuel Gomez-Rodriguez$^{1,2}$}{manuelgr@stanford.edu}

\icmlauthor{Bernhard Sch\"{o}lkopf$^1$}{bs@tuebingen.mpg.de}

\icmladdress{$^1$MPI for Intelligent Systems and $^2$Stanford University}

\icmlkeywords{inferring networks, diffusion}

\vskip 0.3in
]

\begin{abstract}
Diffusion and propagation of information, influence and diseases take place over increasingly larger networks. We observe
when a node copies information, makes a decision or becomes infected but networks are often hidden or unobserved. 
Since networks are highly dynamic, chan\-ging and growing rapidly, we only observe a relatively small set of cascades before
a network changes significantly.
Scalable network inference based on a small cascade set is then necessary for understanding the rapidly evolving dynamics that govern 
diffusion. 
In this article, we develop a scalable approximation algorithm with pro\-va\-ble near-optimal performance based on submodular maximization 
which achieves a high accuracy in such scenario, solving an open problem first introduced by~\citet{manuel10netinf}.
Experiments on syn\-the\-tic and real diffusion data show that our algorithm in practice achieves an optimal trade-off between accu\-ra\-cy 
and running time.
\end{abstract}

\section{Introduction}
\label{sec:intro}
% INTEREST IN DIFFUSION NETWORKS IN DIFFERENT FIELDS
Over the last years, there has been an increasing interest in understanding diffusion and propagation processes in a broad range of domains: information 
propagation~\citep{manuel10netinf}, social networks~\citep{kempe03maximizing}, viral marketing~\citep{dodds07influentials}, 
epidemiology~\citep{wallinga04epidemic},  and human travels~\citep{brockmann2006scaling}.

% EXPLAIN DIFFERENT PROBLEMS: NETWORK INFERENCE, RECONSTRUCTING CASCADES, FINDING END EFFECTORS, MAXIMIZING AND MINIMIZING 
% INFLUENCE
In the context of diffusion networks, one of the fundamental research problems is how to infer the connectivity of a network based on diffusion traces~\citep{manuel10netinf, 
manuel11icml, meyers10netinf, snowsill2011refining}. In information propagation, we note when a blog or news site writes about a piece of 
information. However, in many cases, the blogger or journalist does not link to her source and therefore we do not know where she gathered the information from. 
In viral marketing, we get to know when customers buy products or subscribe to services, but typically cannot observe the \emph{trendsetters} who influenced customers' 
decisions. 
Finally, in epidemiology, we can observe when a person gets ill but cannot tell who infected her.
In all these scenarios, we observe spatiotemporal traces of information spread (be it in the form of a meme, a decision, or a virus) but we do not know the paths 
over which information propagates. We note \emph{where and when} information emerges but not \emph{how or why} it does. In this context, inferring the connectivity 
of diffusion networks is essential to reconstruct and predict the paths over which information spreads, maximize sales of a product or stop infections.

% SCALABILITY PROBLEM

% SMALL NUMBER OF CASCADES WRT NETWORK SIZE (NETWORK CHANGES QUICKLY, SO SMALL NUMBER OF CASCADES FOR A STATIC NETWORK)

% OUR MODIFIED NETINF
\xhdr{Our approach to network inference} We consider that on a fixed hypothetical network, diffusion processes propagate as directed trees through the network. Since we only 
observe the times \emph{when} nodes are reached by a diffusion process, there are many possible propagation trees that explain a set of cascades. 
Naive computation of the model takes exponential time since there is a combinatorially large number of propagation trees. It has been shown that computations over this 
super-exponential set of trees can be performed in cubic time~\citep{manuel10netinf}. However, to the best of our knowledge, efficient optimization of the model has been an open 
question to date. Here, we show that computation over the super-exponential set of trees can indeed be performed in quadratic time and surprisingly, we show that the resulting 
objective function is submodular. Exploiting this natural diminishing property, we can efficiently optimize the objective function to find a near-optimal network with provable guarantees 
that best explain the observed cascades. Lazy evaluation and the local structure of the problem can be used to speed-up our method.
Considering all possible propagation trees enables us to learn a network from fewer observed cascades. This is important since social networks are highly 
dynamic~\citep{backstrom2011supervised}, changing and growing rapidly, and we can only expect to record a small number of cascades over a fixed network.

% RELATED WORK
\xhdr{Related work}
The work most closely related to ours~\citep{manuel10netinf, manuel11icml, meyers10netinf} also uses a generative probabilistic model for inferring diffusion 
networks. \netinf~\citep{manuel10netinf} infers the network connectivity using submodular optimization by considering only the most probable directed tree 
supported by each cascade. \netrate~\citep{manuel11icml} and \connie~\citep{meyers10netinf} infer not only the network connectivity but either prior probabilities of infection 
or transmission rates of infection using convex optimization by considering all possible directed trees supported by each cascade.

% MAIN INNOVATION OF OUR PAPER WITH RESPECT TO RELATED WORK
The main innovation of this paper is to tackle the network inference problem as a submodular maximization problem in which we do not consider only the most probable 
directed tree as in \netinf but all directed trees supported by each cascade as in \connie and \netrate. By considering all trees, we are able to infer a network more 
accurately than \netinf when the number of observed cascades is small compared to the network size and by using the greedy algorithm for submodular maximization 
in contrast to convex optimization, we are several order of magnitude faster than \connie and \netrate. Therefore, we present a network inference algorithm that may be 
capable of inferring real networks in the order of hundred of thousands of nodes with a small number of observed cascades. This comes with a drawback, our algorithm 
does not infer prior probabilities of infection nor transmission rates but only the network connectivity. However, in practice, our algorithm provides a measure of 
\emph{importance} for each edge of the network through the marginal gain that each edge provides.

% WHY WHAT WE DO CAN HELP
Inferring how diffusion propagates over rapidly changing networks is crucial for a better understanding of the dynamics that govern processes taking place over information 
and social networks. In this context, scalability is a key point given the increasingly larger size of such networks and cascade data.

% OUTLINE OF THE PAPER
The remainder of the paper is organized as follows: in Section~\ref{sec:formulation}, we describe the model of diffusion and state the network inference problem. Section~\ref{sec:proposed} 
shows an efficient approximation algorithm with \emph{provable} near-optimal performance. Section~\ref{sec:evaluation} evaluates our method using synthetic and real data and we 
conclude with a discussion of our results in Section~\ref{sec:conclusions}.

\section{Problem formulation}
\label{sec:formulation}
In this section, we first describe the diffusion data our algorithm is designed for and continue revisiting the generative model of diffusion introduced recently by~\citet{manuel10netinf}. 
We conclude with a statement of the network inference problem.
\begin{figure}[t] %
\centering
  \subfigure[Cascade $c$ on $G$]{\includegraphics[width=0.18\textwidth]{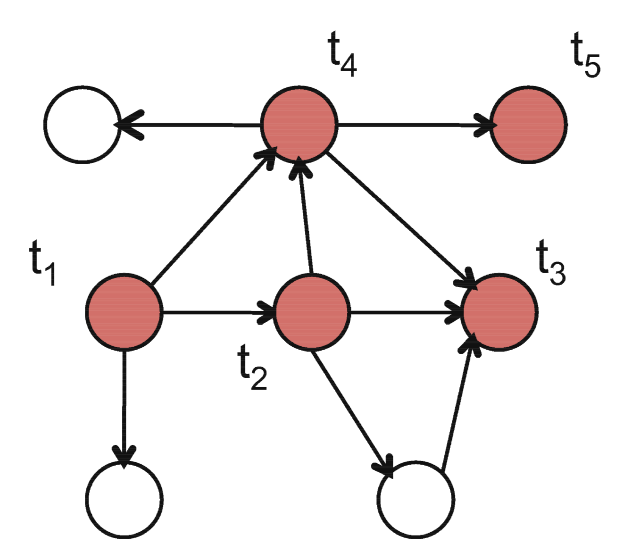} \label{fig:cascade}}
  \subfigure[Spanning trees induced by cascade $c$ on $G$]{\includegraphics[width=0.22\textwidth]{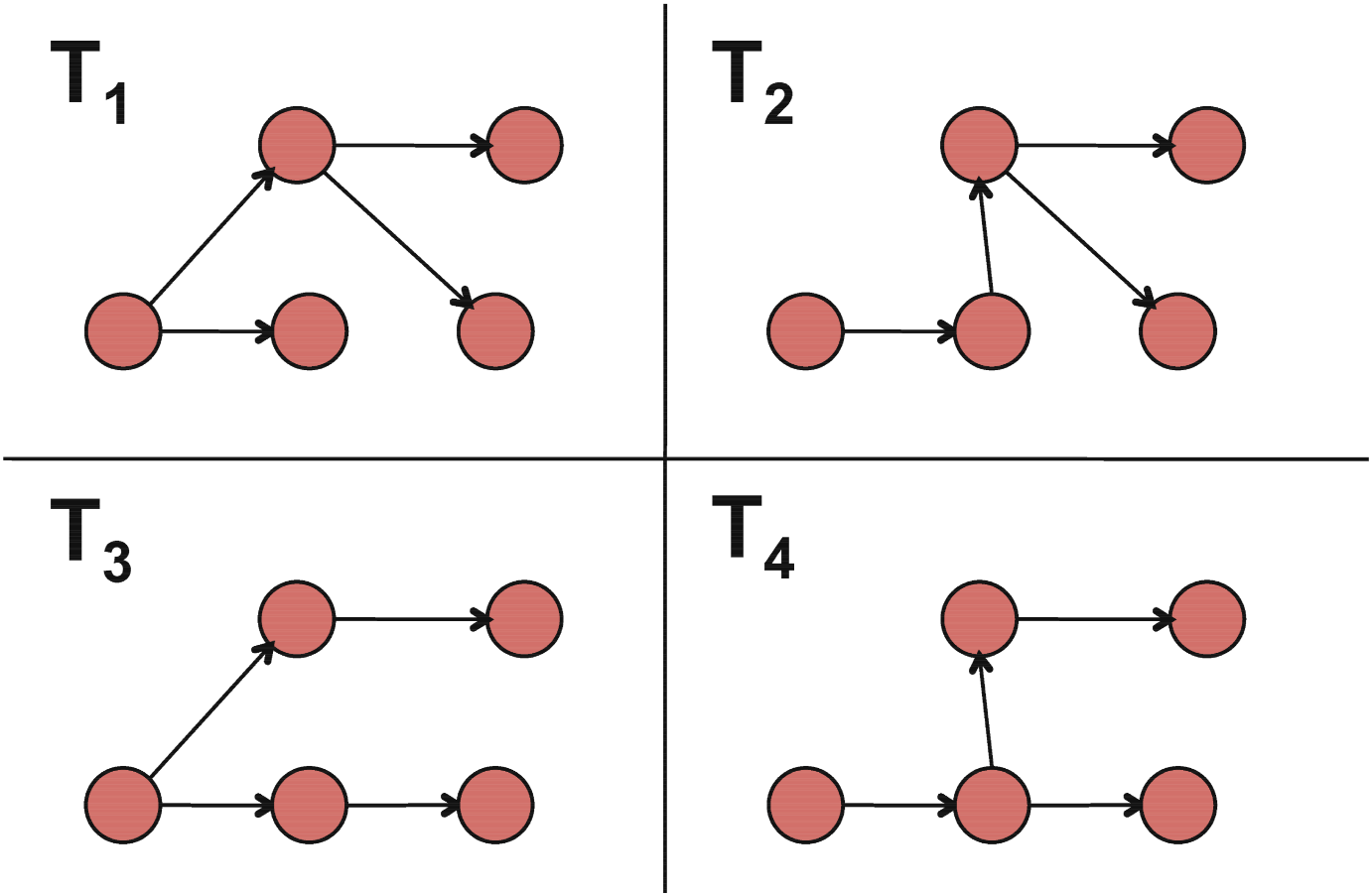} \label{fig:trees}}
  \caption{Panel (a) shows a cascade $\mathbf{t} = \{ t_{1},\dots, t_{5} \}$ on network $G$, where $t_{i-1} < t_{i}$. Panel (b) shows all connected spanning trees induced by 
  cascade $\mathbf{t}$ on $G$, \ie, all possible ways in which a diffusion process spreading over $G$ can create the cascade. \label{fig:illust-cascade-trees}}
\end{figure}

\xhdr{Data} We observe a set $C$ of cascades $\{\bt^{1}, \ldots, \bt^{|C|} \}$ on a fixed population of $N$ nodes. A cascade $\bt^c := (t_1^c, \ldots, t_N^c)$ is simply a N-dimensional 
vector recording when nodes in the population get infected. We only observe the time $t_i^c$ when a node $i$ got infected but not who infected the node neither why it got infected. 
In each cascade, there are typically nodes that are never infected, with infection times that are arbitrarily long. We assume there is an underlying unobserved network $G$ that nodes in 
the population belong to, and cascades propagate over it. Our aim is to discover this unknown network over which cascades originally propagated by using only the recorded infection
times.

\xhdr{Pairwise transmission likelihood} We assume node $j$ can infect node $i$ with prior probability of transmission $\beta$. Now, consider that node $j$ gets infected at 
time $t_j$ and succeeds at infects node $i$ at time $t_i$. We then assume that the infection time $t_i$ depends on $t_j$ through a pairwise transmission likelihood $f(t_i | t_j ; \alpha_{j,i})$. 
As in previous studies of information propagation~\citep{manuel10netinf, manuel11icml} and epidemiology~\citep{wallinga04epidemic}, we consider two 
well-known monotonic parametric models: exponential, $f(t_i | t_j ; \alpha_{j,i}) \propto e^{-\alpha_{j,i} \cdot (t_i - t_j)}$, and power-law, $f(t_i | t_j ; \alpha_{j,i}) \propto (t_i - t_j)^{-1-\alpha_{j,i}}$, 
and one non-monotonic parametric model: Rayleigh, $f(t_i | t_j ; \alpha_{j,i}) \propto (t_i - t_j) e^{-\alpha_{j,i} \cdot (t_i - t_j)^2}$. Although we perform experiments in networks in which the transmission rate $\alpha_{j,i}$ of each edge can be different, in the remainder of the paper, for simplicity, we assume all transmission rates to be equal, $\alpha_{j, i} = \alpha$. Importantly, 
our algorithm does not depend on the particular choice of pairwise transmission likelihood and choosing more complicated parametric or non-parametric likelihoods does not increase its computational complexity.

\xhdr{Likelihood of a cascade for a given tree} We assume that diffusion processes propagate as directed trees, \ie, a node gets infected by action of a single node or parent. Then, 
for a given tree $T$ and cascade $\bt^c$, we can compute the likelihood of the cascade given the tree as follows:
\begin{equation}
f(\bt^c | T) = \prod_{(u,v) \in E_T} f(t_v | t_u ; \alpha),
    \label{eq:ctree}
\end{equation}
where $E_T$ is the edge set of tree $T$. Considering a specific tree $T$ for a cascade $\bt^c$ means to set which edges have spread successfully 
the information. Therefore, given the tree $T$, we can compute the likelihood of the infection times of the nodes in the cascade $\bt^c$ by using simply the pairwise 
transmission likelihood of each edge of the tree.

\xhdr{Probability of a tree in a given network} In order to compute the likelihood of a cascade $\bt^c$ for a given tree $T$, we have considered the tree $T$ to be given. We now compute the 
probability of a tree $T$ in a network $G$ as follows:
\begin{align*}
  P(T | G) &= \prod_{(u,v)\in E_T} \beta \prod_{u \in V_T, (u,x) \in E\backslash E_T} (1-\beta) \\
  &= \beta^{q} (1-\beta)^{r},
\end{align*}
where $V_T$ is the vertex set of tree $T$, $E_T$ is the edge set of tree $T$, $E$ is the edge set of the network $G$ and $q = |E_T| = |V_T|-1$ is the number of edges in $T$ and counts the edges 
over which the diffusion process successfully propagated. For a particular cascade $\bt^c$ and tree $T$, $V_T$ is the set of nodes that belong to $\bt^c$, \ie, nodes where the infection time $t_i < \infty$. 
The first product accounts for the \emph{active} edges in $G$, \ie, edges that define the tree $T$, and the second product accounts for the \emph{inactive} edges in $G$, \ie, edges where the diffusion 
process did not spread. For simplicity, we assume the same prior probability of transmission $\beta$ for every edge of the network $G$.

\xhdr{Likelihood of a cascade in a given network} Now, for a cascade $\bt^c$, we consider all possible propagation trees $T$ that are supported by 
the network $G$, \ie, all possible ways in which a diffusion process spreading over $G$ can create cascade $\bt^c$:
\begin{equation}
  f(\bt^c | G) = \sum_{T \in \mathcal{T}_c(G)} f(\bt^c | T) P(T | G),
    \label{eq:pcasc}
\end{equation}
where $\bt^c$ is a cascade and $\mathcal{T}_c(G)$ is the set of all the directed connected spanning trees on the subnetwork of $G$ induced by the nodes that 
got infected in cascade $\bt^c$, \ie, $t_i \in \bt^c : t_i < \infty$. Figure~\ref{fig:illust-cascade-trees} illustrates the notion of a cascade and all the connected spanning trees $T$ induced by 
its nodes.

All trees $T \in \mathcal{T}_c(G)$ employ the same vertex set $V_T$ and $P(T|G)$ depends only the size of the vertex set $V_T$. Therefore, assuming the same prior probability
of transmission $\beta$ for every edge of the network, $P(T|G)$ is equal for all trees $T$ on the subnetwork of $G$ induced by the nodes that got infected in cascade $\bt^c$ and 
we simplify Eq.~\eqref{eq:pcasc}:
\begin{equation}
  f(\bt^c | G) \propto \sum_{T \in \mathcal{T}_c(G)} \prod_{(u,v)\in E_T} f(t_v | t_u ; \alpha).
    \label{eq:pcasc2}
\end{equation}

Now, assuming conditional independence between cascades given the network $G$, we compute the joint likelihood of a set $C$ of cascades occurring in the network 
$G$ as follows:
\begin{equation}
  f(\bt^1, \ldots, \bt^{|C|}| G) = \prod_{\bt^c \in C} f(\bt^c | G).
  \label{eq:pgraph}
\end{equation}

\xhdr{Network inference problem} Given a set of cascades $\{\bt^1, \ldots, \bt^N\}$, a prior probability of transmission $\beta$ and a pairwise transmission likelihood $f(t_v | t_u ; \alpha)$, 
we aim to find the network $\hat{G}$ such that
\begin{equation}
  \hat G = \argmax_{|G|\leq\nedge} f(\bt^1, \ldots, \bt^N|G),
  \label{eq:pgraph2}
\end{equation}
where the maximization is over all directed networks $G$ of at most $\nedge$ edges.

\section{Proposed algorithm}
\label{sec:proposed}
To the best of our knowledge, the optimization problem defined by Eq.~\eqref{eq:pgraph2} has been considered in\-trac\-ta\-ble in the past and proposed as 
an interesting open problem~\citep{manuel10netinf}. We now show how to efficiently find a solution with \emph{provable} sub-optimality 
guarantees by exploiting a natural diminishing returns property of the network inference problem: submodularity. 

To evaluate Eq.~\eqref{eq:pgraph}, we need to compute Eq.~\eqref{eq:pcasc2} for each cascade $\bt^c$, i.e., compute a sum of likelihoods over all possible 
connected spanning trees $T$ induced by the nodes infected in each cascade. Although the number of trees can be super-exponential in the number of nodes in the 
cascade $\bt^c$, this super-exponential sum can be performed in time polynomial in the number $n$ of nodes in $\bt^c$, by applying Kirchhoff's matrix tree theorem:
\begin{theorem}[Tutte (1948)] \label{thm:tutte}
  Given a directed graph $W$ with non negative edge weights $w_{i,j}$, construct a matrix $A$ such that  $a_{i,j} = \sum_k w_{k, j}$ if $i = j$ and $a_{i,j} = -w_{i,j}$ if $i \neq j$ and denote the 
  matrix created by removing any row $x$ and column $y$ from $A$ as $A_{x,y}$. Then,
  \begin{equation}\label{eq:matrixtree}
  (-1)^{x+y} \det(A_{x,y}) = \sum_{T \in \mathcal{T}(W)} \prod_{(i,j) \in T} w_{i,j},
  \end{equation}
  where $T$ is each directed spanning tree in $W$ that starts in $y$.
\end{theorem}

In our case, we compute Eq.~\eqref{eq:pcasc2} by setting $w_{i,j}$ to $f(t_j | t_i ; \alpha)$ and computing the determinant in Eq.~\eqref{eq:matrixtree}. We then compute 
Eq.~\eqref{eq:pgraph} by multiplying the determinants of $|C|$ matrices, one for each cascade. 
For a fixed cascade $\bt^c$, edges with po\-si\-tive weights form a directed acyclic graph (DAG) (only edges $(i, j)$ such that $t_i < t_j$ have positive weights) and a DAG 
with a time-ordered labeling of its nodes has an upper triangular connectivity matrix. Thus, the matrix $A_{x, y}$ of Theorem~\ref{thm:tutte}, by construction, is also upper 
triangular. Fortunately, the determinant of an upper triangular matrix is simply the product of the elements of its diagonal and then,
$$
f(\bt^c | G) \propto \prod_{t_j \in \bt^c} \sum_{t_i \in \bt^c : t_i \leq t_j} f(t_j | t_i ; \alpha).
$$

This means that instead of using super-exponential time, we are now able to evaluate Eq.~\ref{eq:pgraph} in time $O(|C| \cdot N^2)$, where $N$ is the size of the largest 
cascade, \ie, the time required to build $A_{x, y}$ and compute the determinant for each of the $|C|$ cascades.
\begin{algorithm}[t]
\caption{Our network inference algorithm\label{alg:modified-netinf}}
\begin{algorithmic}
\REQUIRE $C, \nedge$

\STATE $G \leftarrow \bar{K}$;
\WHILE{$|G| < \nedge$}
  \FORALL{$(j,i) \notin G : \exists \bt^c \in C \mbox{ with } t_j < t_i$}
    \STATE $\delta_{j,i} = 0$, $M_{j,i} \leftarrow \emptyset$;
    \FORALL{$\bt^c : t_j < t_i$}
      \STATE $w_c(m,n) \leftarrow$ weight of $(m,n)$ in $G \cup \{(j,i)\}$;
      \FORALL{$t_m : t_m < t_i, m \neq j$}
        \STATE $\delta_{c,j,i} = \delta_{c,j,i} + w_c(m,i)$;
      \ENDFOR
      \STATE $\delta_{j,i} = \log(\delta_{c,j,i}+w_c(j,i))-\log(\delta_{c,j,i}+1$)
    \ENDFOR
  \ENDFOR

  \STATE $(j^{*},i^{*}) \leftarrow \arg \max_{(j,i) \notin G} \delta_{j,i}$;
  \STATE $G \leftarrow G \cup \{(j^{*},i^{*})\}$;
\ENDWHILE

\RETURN G;
\end{algorithmic}
\end{algorithm}

Until now, we have ignored the role of missed infections~\citep{sadikov11cascades} or external sources as mass media~\citep{katz1955personal,dodds07influentials} 
that can produce disconnected cascades. To overcome this point, we consider an additional node $m$ that represents an external source that can infect \emph{any} 
node $u$ in a cascade. Therefore, we connect the external influence source $m$ to every other node $u$ with an $\varepsilon$-edge. Every node $u$ can get infected 
by the external source $m$ with an arbitrarily small probability $\varepsilon$. It is important to remark that adding the external source results in a tradeoff between false 
positives and false negatives when detecting cascades. The higher the value of $\varepsilon$, the larger the number of nodes that are assumed to be infected by an 
external source.

Putting it all together, we include the additional node $m$ in every cascade $\bt^c$ and we set the likelihood of a diffusion process to spread from $m$ to any node $j$ in 
the cascade $\bt^c$ to $\varepsilon$. We assume that $\varepsilon \leq f(t_j | t_i ; \alpha)$ for any $(i, j)$.
We then define the improvement of log-likelihood for cascade $\bt^c$ under graph $G$ over an empty graph $\bar{K}$:
\begin{equation}
  F(\bt^c|G) = \sum_{t_j \in \bt^c} \log \left(\sum_{t_i \in \bt^c : t_i \leq t_j} w_{c}(i,j)\right),
\label{eq:pgraph3}
\end{equation}
where $w_{c}(i,j)= \varepsilon^{-1} f(t_j | t_i ; \alpha) \geq 0$ for all natural likelihoods, $\sum_{i \in G : t_j \geq t_i} w_{c}(i, j) \geq 1$ and we assume that the $\varepsilon$-edges 
between $m$ and all nodes in the cascade $\bt^c$ exist also for the empty graph $\bar{K}$. 

Finally, maximizing Eq.~\eqref{eq:pgraph2} is equivalent to maximizing the following objective function:
\begin{equation}
  F_C(\bt^1, \ldots, \bt^{|C|}|G) = \sum_{\bt^c \in C} F(\bt^c|G),
\end{equation}
where $G$ is the variable.
\begin{figure*}[!!t]
\centering
  \subfigure[PR (Random, \expo)]{\makebox[0.29\textwidth][c]{\includegraphics[width=0.28\textwidth]{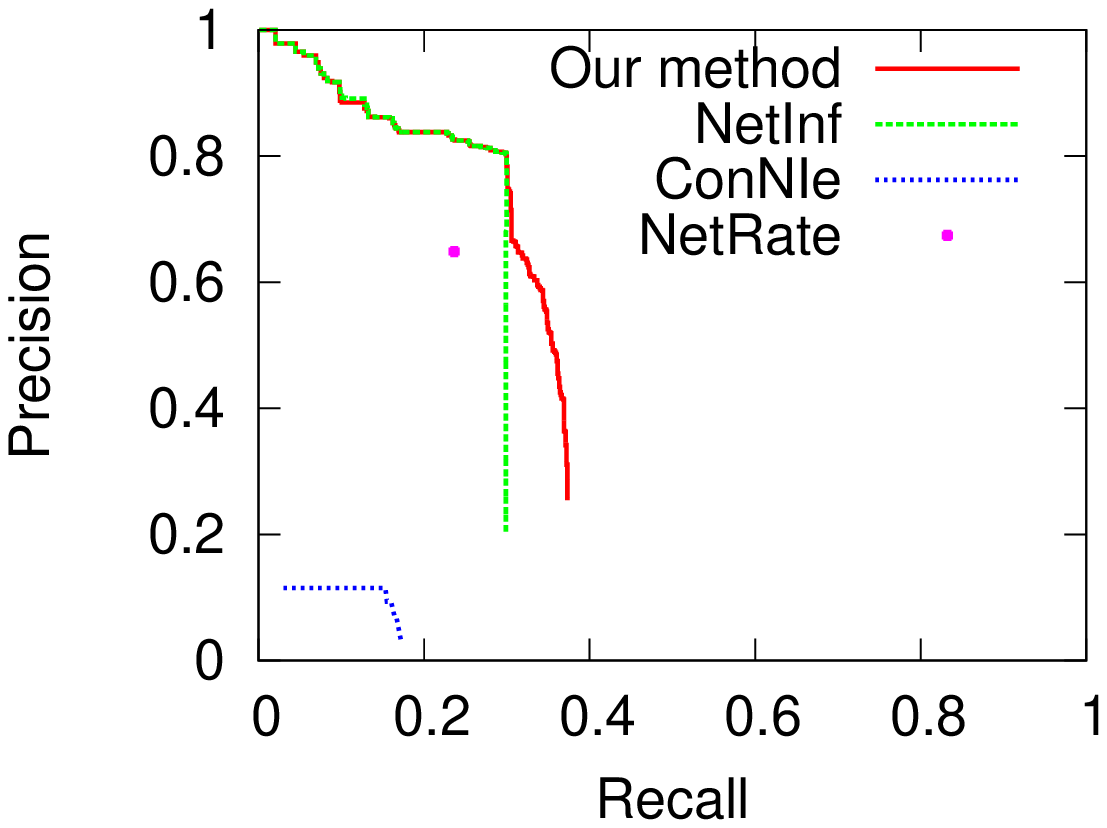}} \label{fig:pr_exp_rnd}} 
  \subfigure[PR (Hierarchical, \pow)]{\makebox[0.29\textwidth][c]{\includegraphics[width=0.28\textwidth]{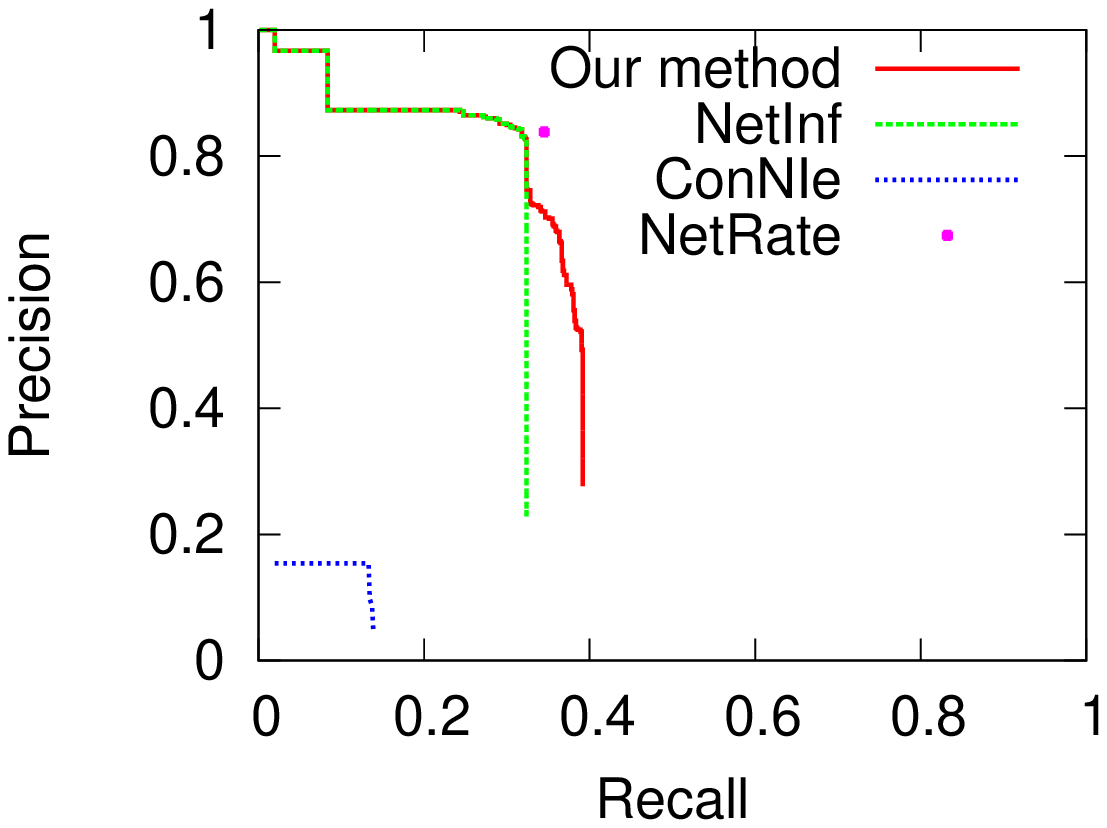}} \label{fig:pr_exp_hie}} 
  \subfigure[PR (Core-periphery, \ray)]{\makebox[0.29\textwidth][c]{\includegraphics[width=0.28\textwidth]{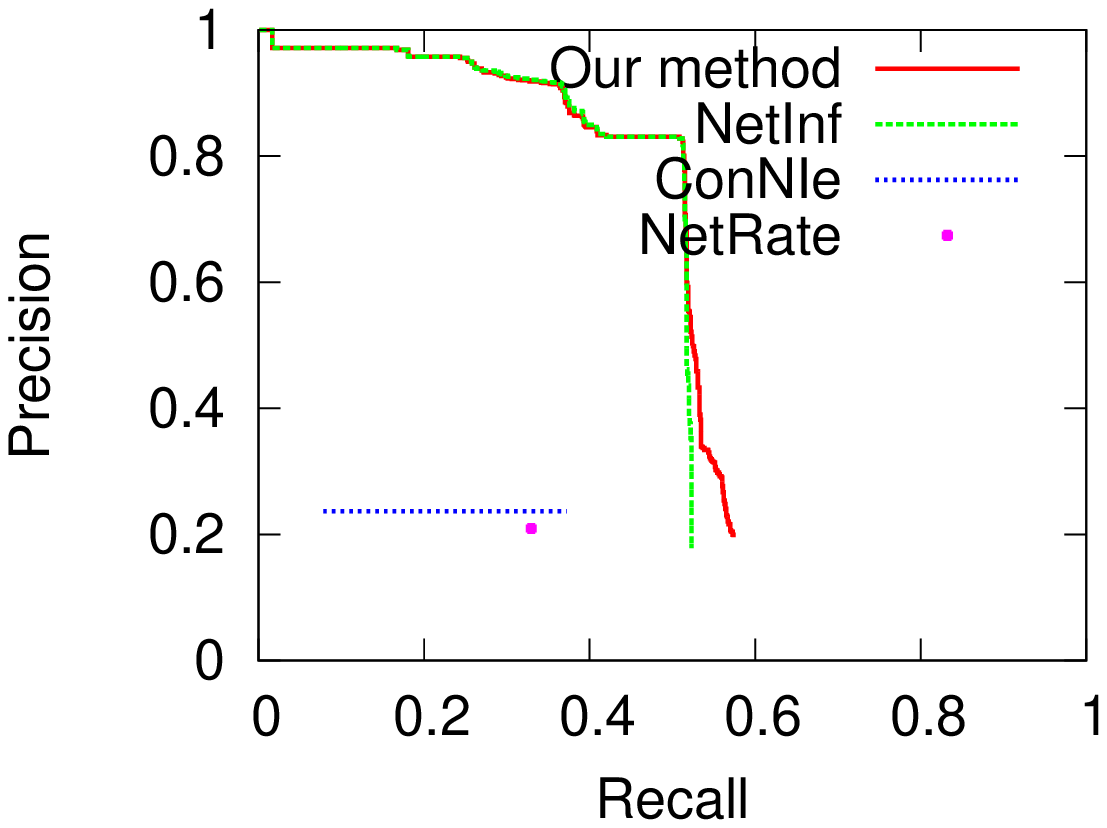}} \label{fig:pr_ray_cp}} \\
  \subfigure[Accuracy (Random, \expo)]{\makebox[0.29\textwidth][c]{\includegraphics[width=0.28\textwidth]{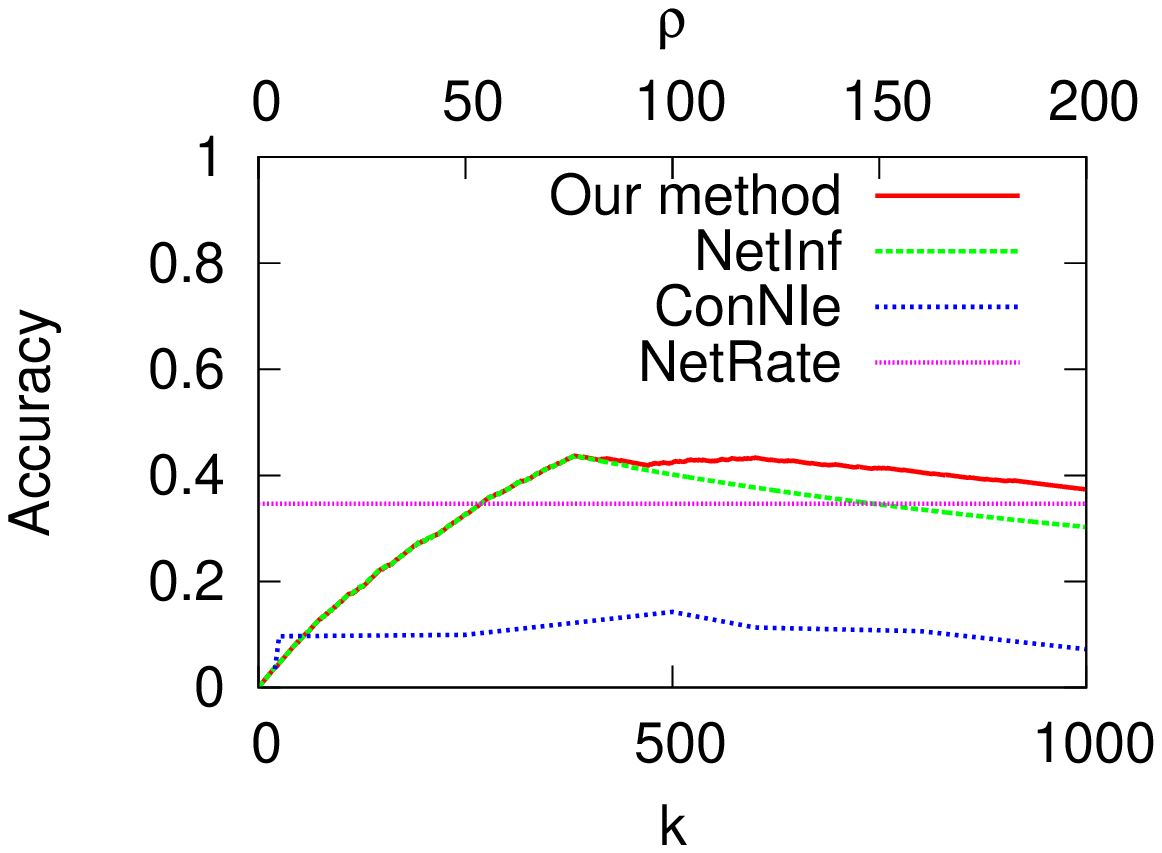}} \label{fig:acc_exp_rnd}} 
  \subfigure[Accuracy (Hierarchical, \pow)]{\makebox[0.29\textwidth][c]{\includegraphics[width=0.28\textwidth]{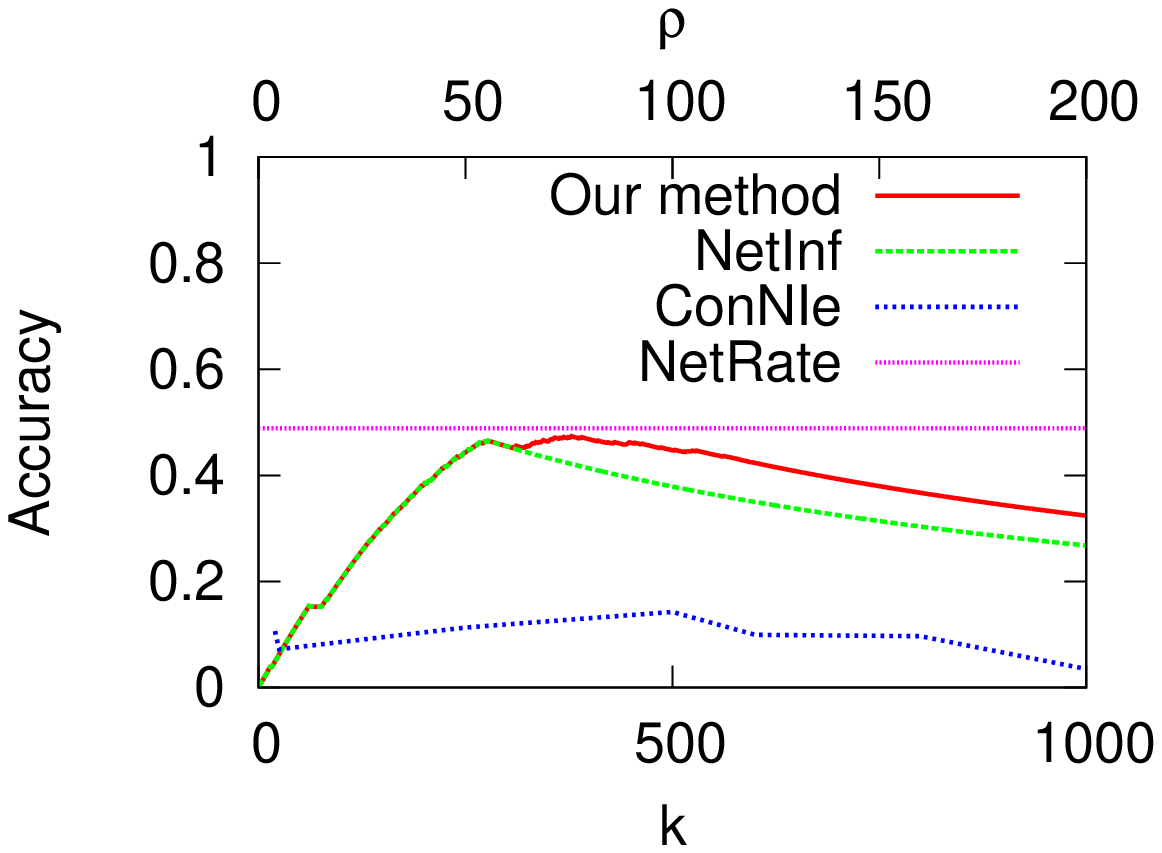}} \label{fig:acc_exp_hie}}
  \subfigure[Accuracy (Core-periphery, \ray)]{\makebox[0.29\textwidth][c]{\includegraphics[width=0.28\textwidth]{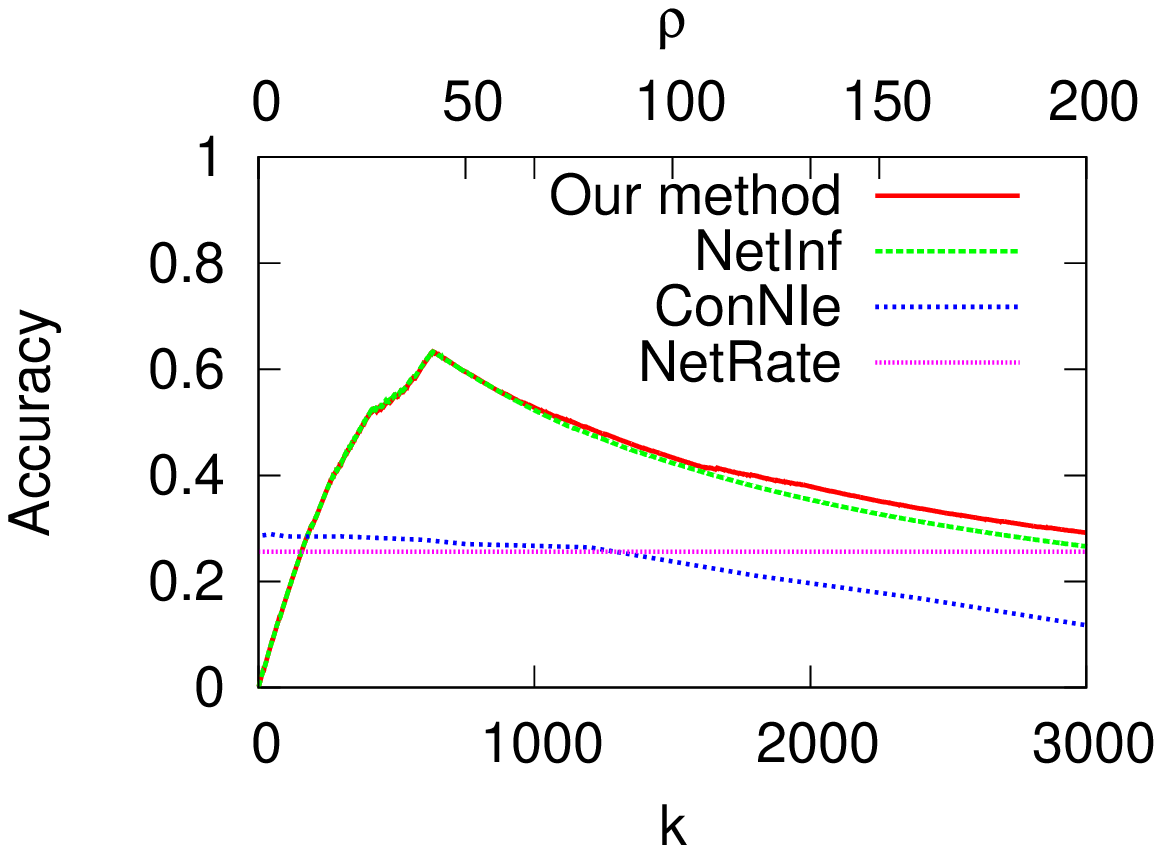}} \label{fig:acc_ray_cp}}
	\caption{
	Panels (a-c) plot precision against recall (PR); panels (d-f) plot accuracy. To control the solution sparsity or precision-recall tradeoff, we sweep over $k$ (number of edges) in our method 
	and \netinf and over $\rho$ (penalty factor) in \connie. \netrate has no tunable parameters and therefore outputs a unique solution.
	(a,d): 1,024 node random Kronecker network with Rayleigh (\ray) model.
	(b,e): 1,024 node hierarchical Kronecker network with power-law (\pow) model.
	(c,f): 1,024 node core-periphery Kronecker network with exponential (\expo) model. In all three networks, we recorded 200 cascades.
	}
	\label{fig:pr}
\end{figure*}

\xhdr{Efficient optimization} By construction, the empty graph $\bar{K}$ has score $0$, $F_{C}(\bt^1, \ldots, \bt^{|C|}|\bar{K})=0$, and the objective function $F_{C}$ is non-negative 
monotonic, $F_{C}(\bt^1, \ldots, \bt^{|C|}|G)\leq F_{C}(\bt^1, \ldots, \bt^{|C|}|G')$, for any $G\subseteq G'$. Therefore, adding more edges to $G$ never decreases the solution quality, 
and thus the complete graph maximizes $F_{C}$.  However, in real-world scenarios, we are interested in inferring sparse graphs with a small number of edges. Thus, we would like to 
solve:
\begin{equation}
    G^{*}=\argmax_{|G|\leq \nedge}F_{C}(\bt^1, \ldots, \bt^{|C|}|G),\label{eq:maxfc}
\end{equation}
where the maximization is over all directed networks $G$ of at most $k$ edges. Naively searching over all $k$ edge graphs would take time exponential in $k$, which is intractable. Moreover, finding the 
optimal solution to Eq.~\ref{eq:maxfc} is NP-hard:
\begin{theorem} \label{thm:np}
The diffusion network inference problem defined by Eq.~\ref{eq:maxfc} is NP-hard.
\end{theorem}
\begin{proof}
By reduction from the MAX-$k$-COVER pro\-blem~\cite{khuller1999budgeted}.
\end{proof}

While finding the optimal solution is hard, we will now show that $F_{C}$ satisfies \emph{submodularity} on the set of directed edges in $G$, a natural diminishing returns property, which 
will allow us to efficiently find a \emph{provable} near-optimal solution to the optimization problem.

A set function $F: 2^{W}\rightarrow\mathbb{R}$ mapping subsets of a finite set $W$ to the real numbers is \emph{submodular} if whenever $A\subseteq B\subseteq
W$ and $s\in W\setminus B$, it holds that $F(A\cup\{s\})-F(A)\geq F(B\cup\{s\})-F(B)$, i.e., adding $s$ to the set $A$ increases the score more than adding $s$ to the 
set $B$.  We have the following result:
\begin{theorem}\label{thm:submodular}
Let $V$ be a set of nodes, and $C$ be a collection of cascades hitting the nodes $V$. Then $F_{C}(\bt^1, \ldots, \bt^{|C|}|G)$ is a submodular function 
$F_{C} : 2^{W}\rightarrow \mathbb{R}$ defined over subsets $W\subseteq V\times V$ of directed edges.
\end{theorem}
\begin{proof}
Fix a cascade $\bt^c$, graphs $G\subseteq G'$ and an edge $e=(r,s)$ not contained in $G'$.  We will show that $F(\bt^c|G\cup\{e\})-F(\bt^c|G)\geq
F(\bt^c|G'\cup\{e\})-F(\bt^c|G')$. 
Let $w_{i,j}$ be the weight of edge $(i,j)$ in $G$, and $w'_{i,j}$ in $G'$. Since $G\subseteq G'$, it holds that $w'_{i,j} \geq w_{i,j} \geq 0$. If $(i,j)$ is contained 
in $G$ and $G'$, then $w_{i,j} = w'_{i,j}$. Let $T_{A, e} = \sum_{i \in A \backslash \{r\} : t_j \geq t_i} w_{c}(i, s)$. It holds that $T_{G', e} \geq T_{G, e}$. 
Hence,
\begin{align*}F(\bt^c|G\cup\{e\})-F(\bt^c|G)&=\log \left( \frac{T_{G, e} + w_{c}(r, s)}{T_{G, e}} \right) \\
&\geq \log \left( \frac{T_{G', e} + w_{c}(r, s)}{T_{G', e}} \right) \\
&=F(\bt^c|G'\cup\{e\})-F(\bt^c|G'),
\end{align*}
proving submodularity of $F(\bt^c|G)$. Now, since nonnegative linear combinations of submodular functions are submodular, the function 
\begin{align*}
F_{C}(\bt^1, \ldots, \bt^{|C|}|G)=\sum_{c\in C}F(\bt^c|G)
\end{align*}
is submodular as well. 
\end{proof}

We now optimize $F_{C}(G)$ by using the \emph{greedy algorithm}, a well-known efficient heuristic with provable performance guarantees. The algorithm starts 
with an empty graph $\bar{K}$ and it adds edges that maximize the \emph{marginal gain} sequentially. That means, at iteration $i$ we choose the edge
$
e_i = \argmax_{e \in G \backslash G_{i-1}} F_C(G_{i-1} \cup \{e\}) - F_C(G_{i-1}).
$

The algorithm stops once it has selected $\nedge$ edges, and returns the solution $\hat{G}=\{e_{1},\dots,e_{\nedge}\}$. The greedy algorithm is guaranteed to 
find a set $\hat{G}$ which achieves at least a constant fraction $(1-1/e)$ (of the optimal value achievable using $\nedge$ edges~\citep{nemhauser1978analysis}.
Starting from the near-optimal solution given by the greedy algorithm, we could possibly improve the solution by applying a local search procedure.
\begin{figure*}[t]
	\centering
	\subfigure[Random]{\makebox[0.29\textwidth][c]{\includegraphics[width=0.28\textwidth]{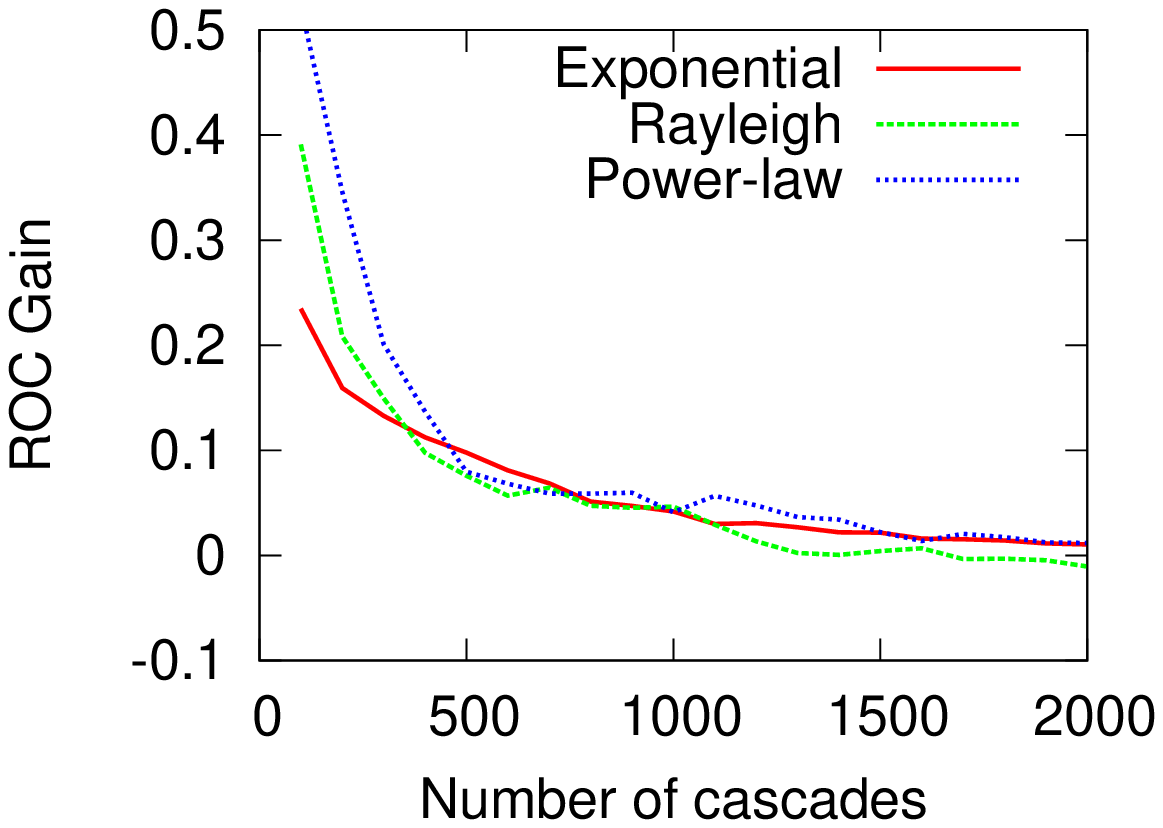}}
	\label{fig:performance-vs-num-cascades-1}}
	\subfigure[Hierarchical]{\makebox[0.29\textwidth][c]{\includegraphics[width=0.28\textwidth]{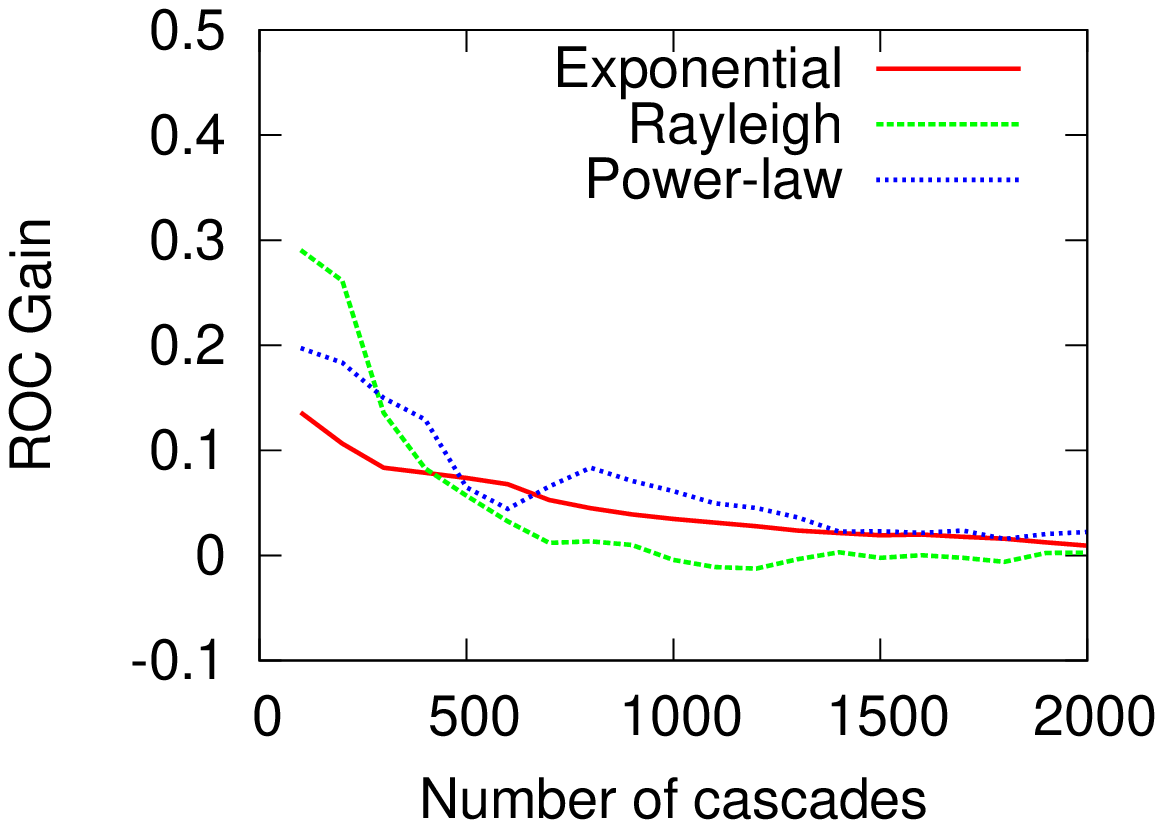}}
	\label{fig:performance-vs-num-cascades-2}}
	\subfigure[Core-periphery]{\makebox[0.29\textwidth][c]{\includegraphics[width=0.28\textwidth]{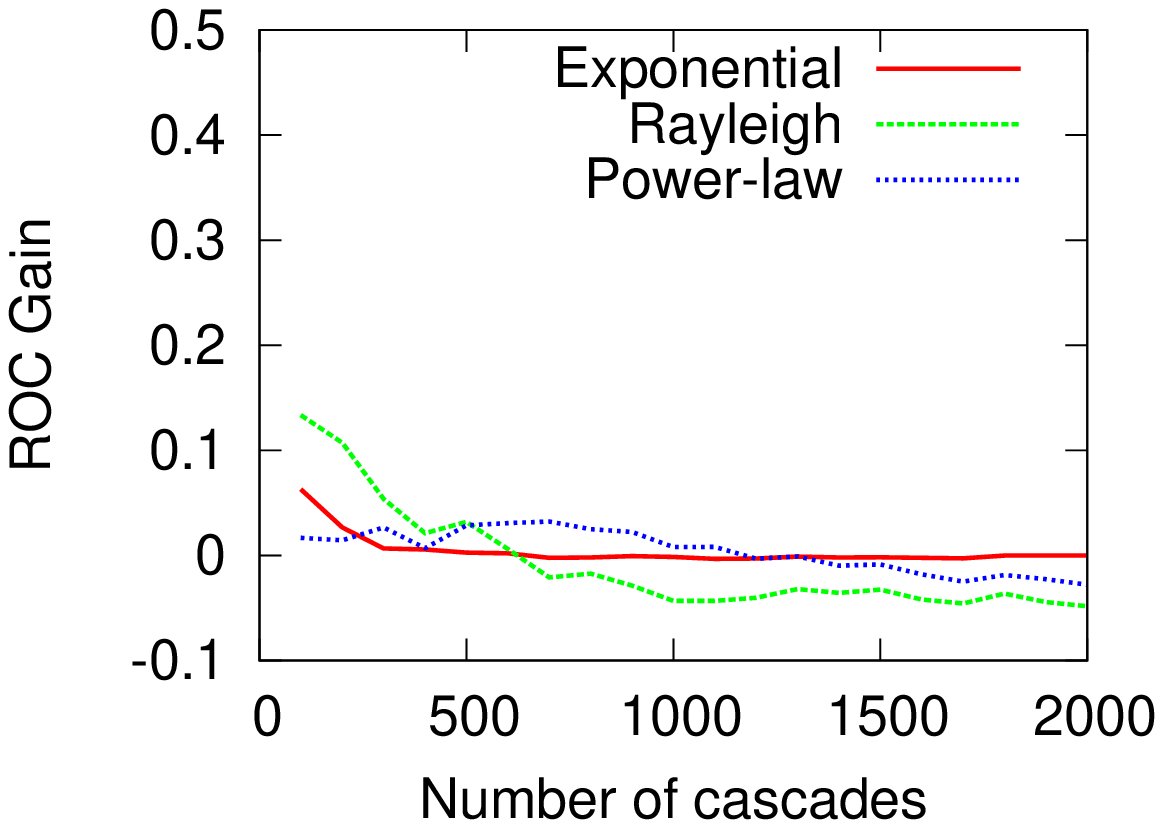}}
	\label{fig:performance-vs-num-cascades-3}}
 	\caption{Gain in Area Under the ROC curve (AUC) of our method compared to \netinf vs number of cascades for (a) a random Kronecker network,
	(b) a hierarchical Kronecker network and (c) a core-periphery Kronecker network with 1,024 nodes and 1,024 edges for all three transmission models. 
	Our method is able to more accurately infer a network for small number of cascades and it exhibits similar performance to \netinf for larger number of 
	cascades.
	} \label{fig:performance-vs-num-cascades}
\end{figure*}

As in the original \netinf formulation, our algorithm also allows for two speeds-up: localized updates and lazy evaluation (Algorithm~\ref{alg:modified-netinf}). We can also 
obtain an on-line bound based simply on the submodularity of the objective function~\cite{leskovec2007cost}.

\section{Experimental evaluation}
\label{sec:evaluation}
We evaluate our network inference algorithm in both synthetic and real networks. We use synthetic networks that aim to mimic the structure of
social networks, and real information networks that are based on the MemeTracker dataset\footnote{Data available at \url{http://memetracker.org}}.
We compare our method in terms of precision, recall, accuracy and scalability with several state-of-the-art algorithms: \netinf, \connie and \netrate. For 
the comparisons, we use the public domain implementations of these algorithms.

\subsection{Experiments on synthetic data}

\xhdr{Experimental setup}
We first generate synthetic networks using two different well-known models of social networks: the Forest Fire (scale free) model~\citep{barabasi99emergence} and the 
Kronecker model~\citep{leskovec2010kronecker}, and set the pairwise transmission rates of the edges of the networks by drawing samples from $\alpha \sim U(0.5, 1.5)$.
We then simulate and record a relatively small set of propagating cascades over each network using three different pairwise transmission likelihoods: exponential, power-law and 
Rayleigh. There are several reasons why we consider small set of cascades in comparison to the network size. First, all methods (including ours) assume that cascades propagate 
over a fixed network. Since social networks are highly dynamic~\citep{backstrom2011supervised}, changing and growing rapidly, we can only expect to record a small number of 
cascades over a fixed network. Second, tracking and recording cascades is a difficult and expensive process~\citep{leskovec2009kdd}. Therefore, it is desirable to develop network 
inference methods that work well with a small number of observed cascades.

\xhdr{Accuracy}
We compare the inferred and true networks via three measures: precision, recall and accuracy. Precision is the fraction of edges in the inferred network $\hat{G}$ present in the 
true network $G^*$ . Recall is the fraction of edges of the true network $G^*$ present in the inferred network $\hat{G}$. Accuracy is $1-\frac{\sum_{i,j} |I(\alpha^*_{i,j})-I(\hat{\alpha}_{i,j})|}{\sum_{i,j}I(\alpha^*_{i,j}) + \sum_{i,j} I(\hat{\alpha}_{i,j})}$, where $I(\alpha)=1$ if $\alpha > 0$ and $I(\alpha)=0$ otherwise. Inferred networks with no edges or only false edges have zero accuracy.

Figure~\ref{fig:pr} compares our method  the precision, recall and accuracy of our method with for three different 1,024 node Kronecker networks: a random 
network~\citep{erdos60random} (parameter matrix $[0.5, 0.5; 0.5, 0.5]$), a hierarchical network~\citep{clauset08hierarchical} ($[0.9, 0.1; 0.1, 0.9]$) and a core-periphery 
network~\citep{jure08ncp} ($[0.9, 0.5; 0.5, 0.3]$), and $200$ observed cascades.
In terms of precision-recall, our method is able to reach higher recall values than \netinf, \connie and \netrate, \ie, it is able to discover more true edges from a small number of cascades 
than other methods. For recall values that are reachable using \netinf, our method and \netinf offer a very similar precision value. Our methods allows for higher recall in comparison with
\netinf because it gets exhausted\footnote{A greedy method (ours and \netinf) gets exhausted at iteration $k$ when there are not any more edges with marginal gain larger than zero.} later 
for considering all possible trees per cascade instead of only the most probable one. 
In terms of accuracy, our method outperforms \netinf for more than half of their outputted solutions, and matches the remaining ones. \connie and \netrate'{}s accuracy is typically
significantly lower. However, \netrate is able to beat all other methods for the hierarchical Kronecker network.
If we compare with previous studies~\cite{meyers10netinf}, the performance of \connie seem to degrade the most due to the limited availability in cascades and perhaps the variable
transmission rates across the networks (as reported previously in \citet{manuel11icml}).

\xhdr{Performance vs. cascade coverage} Intuitively, the more cascades we observe, the more accurately \emph{any} algorithm infers a network. Actually, when the number of 
cascades is large in comparison to the network size, we expect differences in performance among methods become negligible. 
Figure~\ref{fig:performance-vs-num-cascades} plots the gain in Area Under the ROC curve (AUC) for our method in comparison with \netinf, (AUC$_{\mbox{\tiny our method}} -
$ AUC$_{\mbox{\tiny \netinf}}$)/AUC$_{\mbox{\tiny \netinf}}$, against number of observed cascades for several Kronecker networks and 
transmission models ($\beta = 0.5$ and $\alpha \sim U(0.5, 1.5)$ in all models). We observe that the difference in performance between our method and \netinf is greater for 
small number of cascades and for a large enough number of cascades, both methods perform similarly or \netinf slightly outperforms our method. 
\begin{figure}[t]
	\centering
	\includegraphics[width=0.28\textwidth]{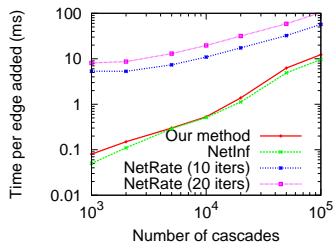} 	
	\caption{Average running time per edge added against number of cascades. We used a 1,024 node random Kronecker with exponential transmission model.} 
	\label{fig:scalability}
\end{figure}

\xhdr{Scalability} Figure~\ref{fig:scalability} plots the average computation time per edge added against number of cascades. 
Since \netrate is not greedy and instead solve a convex program for each node in the network, we divided their total running times by the number of edges that our method added until 
getting exhausted (until no edge has marginal gain greater than zero).
We used the publicly available implementations of our algorithm and \netinf, both coded in C++. To carry out a fair comparison with \netrate, we have developed a projected full gradient 
descend C++ implementation of \netrate, which is considerably faster than the publicly available Matlab implementation (that uses the CVX convex solver), and we run $10$ and $20$
iterations of full gradient descend (remarkably, even running one single iteration was slower than \netinf and our method). We do not report running times for \connie since the publicly 
available code is a Matlab implementation (that uses the SNOPT solver) and probably slower than a C++ implementation.
Our method and \netinf are approximately one order of magnitude faster than \netrate. 
Finally, note that the running time of our algorithm does not depend on the network size but the number of cascades and cascade size. As an experimental validation, we run our algorithm 
in two networks with $100,000$ and $200,000$ nodes and an average of two edges per node using $10,000$ cascades and our algorithm took only $10.12$ ms and $12.14$ ms per edge added.

\subsection{Experiments on real data}

\xhdr{Experimental setup} We use the publicly available MemeTracker dataset, which contains more than $172$ million news articles and blog posts from $1$ million online 
sources~\citep{leskovec2009kdd}. 
Sites publish pieces of information and use hyperlinks to refer to their sources, which are other sites that published the same or closely related pieces of information. Therefore, 
we use hyperlinks to trace information propagation over blogs and media sites.
A hyperlink cascade is simply a collection of time-stamped hyperlinks between sites (in blog or news media posts) that refer to the same or closely related pieces of information. We record 
one hyperlink cascade per piece or closely related pieces of information. We extract the top 1,000 media sites and blogs with the largest number of documents, 10,000 hyperlinks and 
500 longest hyperlink cascades.
We create a ground truth network $G$ which contains an edge between a site $u$ and a site $v$ if there is at least a site post in the site $u$ that links to a post on the site $v$. We then
infer a network $\hat{G}$ from the hyperlink cascades and evaluate precision, recall and accuracy with respect to $G$. We consider a power law pairwise transmission likelihood.
Note that we trace the flow of information and create a ground truth network using hyperlinks because we are interested in a quantitative evaluation of our method in comparison with the 
state of the art. For richer qualitative insights, cascades based on short textual phrases should be considered, but that goes beyond the scope of this paper.

\xhdr{Accuracy} Figure~\ref{fig:performance-real-data} shows precision, recall and accuracy of our method in comparison with \netinf, \connie and \netrate. Our method reaches higher 
recall values than any other methods. In terms of accuracy, it beats others for the majority of their outputted solutions. As in the synthetic experiments, the shortage of recorded cascades 
degrades \connie'{}s performance dramatically.
\begin{figure}[t]
	\centering
	\subfigure[Precision-recall]{\includegraphics[width=0.23\textwidth]{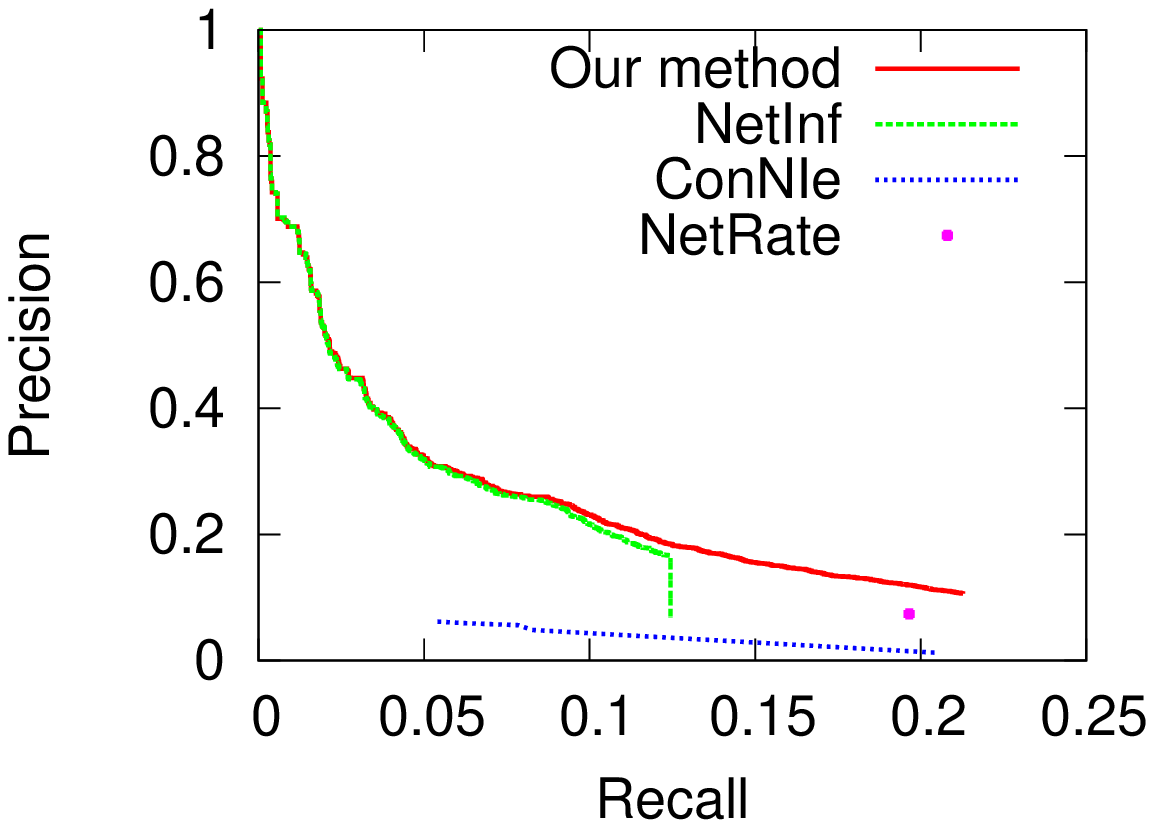} 
	\label{fig:precision-recall-real-data}}
	\subfigure[Accuracy]{\includegraphics[width=0.23\textwidth]{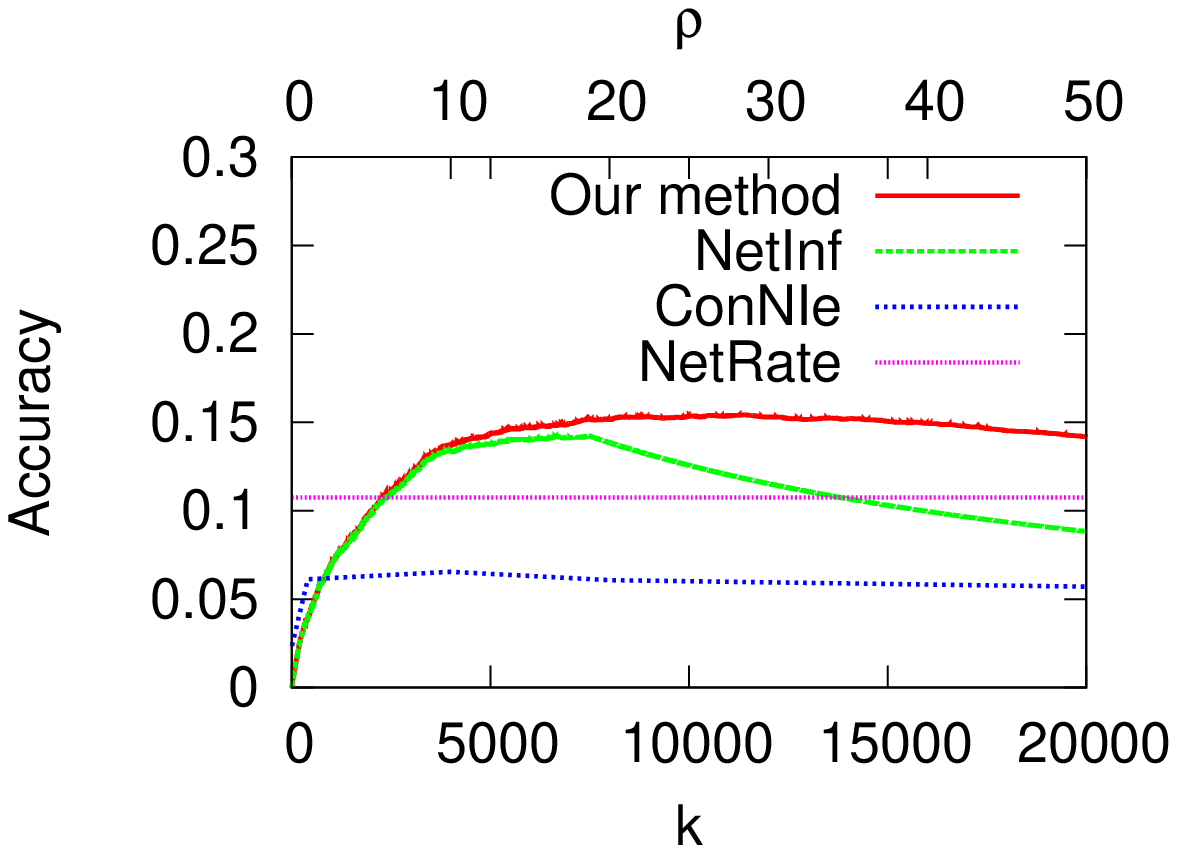}
	\label{fig:accuracy-real-data}} 	
	\caption{Real data. Panel (a) plots precision-recall and panel (b) accuracy on a 1,000 node hyperlink network with 10,000 edges using 1,000 cascades and a power-law model.
	To control the solution sparsity or precision-recall tradeoff, we sweep over k (number of edges) in our method and \netinf and over $\rho$ (penalty factor) in \connie. Our method 
	beats others for the majority of their outputted solutions. 
	} 
	\label{fig:performance-real-data}
\end{figure}

\section{Conclusions}
\label{sec:conclusions}
We have developed an efficient approximation algorithm with pro\-va\-ble near-optimal performance that solves an open problem on network inference from diffusion traces (or 
cascades) first introduced by~\citet{manuel10netinf}. 
In our work, for each observed cascade we consider all possible ways in which a diffusion process spreading over the network can create the cascade, in contrast with \netinf, that 
considers only the most probable way (tree). 

Perhaps surprisingly, despite considering all trees, we show experimentally that the running time of our method and \netinf are similar, and they are several orders of magnitude faster 
than alternative network inference methods based on convex programming as \netrate and \connie. Moreover, our algorithm typically outperforms \netinf, \netrate and \connie in terms 
of precision, recall and accuracy in highly dynamic networks in which we only observe a relatively small set of cascades before they change significantly.

\bibliographystyle{icml2012}
\bibliography{refs}

\begin{thebibliography}{20}
\providecommand{\natexlab}[1]{#1}
\providecommand{\url}[1]{\texttt{#1}}
\expandafter\ifx\csname urlstyle\endcsname\relax
  \providecommand{\doi}[1]{doi: #1}\else
  \providecommand{\doi}{doi: \begingroup \urlstyle{rm}\Url}\fi

\bibitem[Backstrom \& Leskovec(2011)Backstrom and
  Leskovec]{backstrom2011supervised}
Backstrom, L. and Leskovec, J.
\newblock Supervised random walks: predicting and recommending links in social
  networks.
\newblock In \emph{WSDM '11}, 2011.

\bibitem[Barab\'{a}si \& Albert(1999)Barab\'{a}si and
  Albert]{barabasi99emergence}
Barab\'{a}si, A.-L. and Albert, R.
\newblock Emergence of scaling in random networks.
\newblock \emph{Science}, 286:\penalty0 509--512, 1999.

\bibitem[Brockmann et~al.(2006)Brockmann, Hufnagel, and
  Geisel]{brockmann2006scaling}
Brockmann, D., Hufnagel, L., and Geisel, T.
\newblock {The scaling laws of human travel}.
\newblock \emph{Nature}, 439\penalty0 (7075):\penalty0 462--465, 2006.
\newblock ISSN 0028-0836.

\bibitem[Clauset et~al.(2008)Clauset, Moore, and Newman]{clauset08hierarchical}
Clauset, A., Moore, C., and Newman, M. E.~J.
\newblock Hierarchical structure and the prediction of missing links in
  networks.
\newblock \emph{Nature}, 453\penalty0 (7191):\penalty0 98--101, 2008.

\bibitem[Erd\H{o}s \& R\'{e}nyi(1960)Erd\H{o}s and R\'{e}nyi]{erdos60random}
Erd\H{o}s, P. and R\'{e}nyi, A.
\newblock On the evolution of random graphs.
\newblock \emph{Publication of the Mathematical Institute of the Hungarian
  Academy of Science}, 5:\penalty0 17--67, 1960.

\bibitem[Gomez-Rodriguez et~al.(2010)Gomez-Rodriguez, Leskovec, and
  Krause]{manuel10netinf}
Gomez-Rodriguez, M., Leskovec, J., and Krause, A.
\newblock {Inferring Networks of Diffusion and Influence}.
\newblock In \emph{KDD '10: Proceedings of the 16th ACM SIGKDD International
  Conference on Knowledge Discovery in Data Mining}, pp.\  1019--1028, 2010.

\bibitem[Gomez-Rodriguez et~al.(2011)Gomez-Rodriguez, Balduzzi, and
  Sch\"{o}lkopf]{manuel11icml}
Gomez-Rodriguez, M., Balduzzi, D., and Sch\"{o}lkopf, B.
\newblock {Uncovering the Temporal Dynamics of Diffusion Networks}.
\newblock In \emph{ICML '11}, 2011.

\bibitem[Katz \& Lazarsfeld(1955)Katz and Lazarsfeld]{katz1955personal}
Katz, E. and Lazarsfeld, P.F.
\newblock \emph{{Personal influence: The part played by people in the flow of
  mass communications}}.
\newblock Free Press, 1955.

\bibitem[Kempe et~al.(2003)Kempe, Kleinberg, and Tardos]{kempe03maximizing}
Kempe, D., Kleinberg, J.~M., and Tardos, \'{E}.
\newblock Maximizing the spread of influence through a social network.
\newblock In \emph{KDD '03: Proceedings of the 9th ACM SIGKDD International
  Conference on Knowledge Discovery and Data Mining}, pp.\  137--146, 2003.

\bibitem[Khuller et~al.(1999)Khuller, Moss, and Naor]{khuller1999budgeted}
Khuller, S., Moss, A., and Naor, J.
\newblock {The budgeted maximum coverage problem}.
\newblock \emph{Information Processing Letters}, 70\penalty0 (1):\penalty0
  39--45, 1999.

\bibitem[Leskovec et~al.(2007)Leskovec, Krause, Guestrin, Faloutsos,
  VanBriesen, and Glance]{leskovec2007cost}
Leskovec, J., Krause, A., Guestrin, C., Faloutsos, C., VanBriesen, J., and
  Glance, N.
\newblock {Cost-effective outbreak detection in networks}.
\newblock In \emph{KDD '07: Proceedings of the 13th ACM SIGKDD international
  conference on Knowledge discovery and data mining}, pp.\  420--429, 2007.

\bibitem[Leskovec et~al.(2008)Leskovec, Lang, Dasgupta, and Mahoney]{jure08ncp}
Leskovec, J., Lang, K.~J., Dasgupta, A., and Mahoney, M.~W.
\newblock Statistical properties of community structure in large social and
  information networks.
\newblock In \emph{WWW '08}, 2008.

\bibitem[Leskovec et~al.(2009)Leskovec, Backstrom, and
  Kleinberg]{leskovec2009kdd}
Leskovec, J., Backstrom, L., and Kleinberg, J.
\newblock Meme-tracking and the dynamics of the news cycle.
\newblock In \emph{KDD '09: Proceedings of the 15th ACM SIGKDD International
  Conference on Knowledge Discovery and Data Mining}, 2009.

\bibitem[Leskovec et~al.(2010)Leskovec, Chakrabarti, Kleinberg, Faloutsos, and
  Ghahramani]{leskovec2010kronecker}
Leskovec, J., Chakrabarti, D., Kleinberg, J., Faloutsos, C., and Ghahramani, Z.
\newblock {Kronecker graphs: An approach to modeling networks}.
\newblock \emph{The Journal of Machine Learning Research}, 11:\penalty0
  985--1042, 2010.

\bibitem[Myers \& Leskovec(2010)Myers and Leskovec]{meyers10netinf}
Myers, S. and Leskovec, J.
\newblock {On the Convexity of Latent Social Network Inference}.
\newblock In \emph{NIPS '10}, 2010.

\bibitem[Nemhauser et~al.(1978)Nemhauser, Wolsey, and
  Fisher]{nemhauser1978analysis}
Nemhauser, GL, Wolsey, LA, and Fisher, ML.
\newblock {An analysis of approximations for maximizing submodular set
  functions}.
\newblock \emph{Mathematical Programming}, 14\penalty0 (1):\penalty0 265--294,
  1978.

\bibitem[Sadikov et~al.(2011)Sadikov, Medina, Leskovec, and
  Garcia-Molina]{sadikov11cascades}
Sadikov, S., Medina, M., Leskovec, J., and Garcia-Molina, H.
\newblock Correcting for missing data in information cascades.
\newblock In \emph{WSDM '11: ACM International Conference on Web Search and
  Data Mining}, 2011.

\bibitem[Snowsill et~al.(2011)Snowsill, Fyson, De~Bie, and
  Cristianini]{snowsill2011refining}
Snowsill, T.M., Fyson, N., De~Bie, T., and Cristianini, N.
\newblock Refining causality: who copied from whom?
\newblock In \emph{KDD '11: Proceedings of the 17th ACM SIGKDD international
  conference on Knowledge discovery and data mining}, pp.\  466--474, 2011.

\bibitem[Wallinga \& Teunis(2004)Wallinga and Teunis]{wallinga04epidemic}
Wallinga, J. and Teunis, P.
\newblock Different epidemic curves for severe acute respiratory syndrome
  reveal similar impacts of control measures.
\newblock \emph{American Journal of Epidemiology}, 160\penalty0 (6):\penalty0
  509--516, 2004.

\bibitem[Watts \& Dodds(2007)Watts and Dodds]{dodds07influentials}
Watts, Duncan~J. and Dodds, Peter~S.
\newblock Influentials, networks, and public opinion formation.
\newblock \emph{Journal of Consumer Research}, 34\penalty0 (4):\penalty0
  441--458, 2007.

\end{thebibliography}

\end{document}